\documentclass[11pt,letterpaper,oneside,english]{article}
\usepackage[left=3.5cm,right=3.5cm,top=3cm,bottom=3cm]{geometry}

\usepackage{mdwlist,enumerate}
\usepackage{datetime}

\makeatletter
\let\@font@warningori\@font@warning
\newcommand\shutup{\def\@font@warning##1{}}
\newcommand\youcanspeak{\let\@font@warning\@font@warningori}
\makeatother 
\usepackage{amsmath,amsbsy,bm,latexsym,amssymb}
\shutup
\usepackage{fourier}
\youcanspeak
\usepackage[scaled=0.875]{helvet}

\usepackage{graphics, subfigure, float}
\usepackage{fp, calc}

\usepackage[latin1]{inputenc}
\usepackage[american]{babel}
\usepackage[T1]{fontenc} 

\usepackage{amscd,amsthm}

\usepackage{verbatim, comment}

\usepackage[dvips]{graphicx,epsfig,color}
\usepackage{pst-all}
\usepackage{pstricks-add}

\theoremstyle{theorem}
\newtheorem{theorem}{Theorem}
\newtheorem*{theorem*}{Theorem}
\newtheorem{corollary}[theorem]{Corollary}

\newtheorem{lemma}[theorem]{Lemma}
\newtheorem*{lemma*}{Lemma}

\newtheorem*{claim*}{Claim}

\newtheorem*{conjecture*}{Conjecture}
\newtheorem*{problem*}{Problem}

\newtheorem*{definition*}{Definition}
\newtheorem{observation}{Observation}

\theoremstyle{remark}

\newtheorem*{remark*}{Remark}

\newtheorem*{algorithm*}{Algorithm}

\providecommand{\setN}{\mathbb{N}}
\providecommand{\setZ}{\mathbb{Z}}
\providecommand{\setQ}{\mathbb{Q}}
\providecommand{\setR}{\mathbb{R}}

\newcommand{\supp}{\textrm{supp}}

\newcommand{\polylog}{\ensuremath{\textrm{polylog}}}

\floatstyle{ruled}
\newfloat{algorithm}{tbp}{loa}
\floatname{algorithm}{Algorithm}
\newfloat{algorithm*}{tbp}{loa}
\floatname{algorithm*}{Algorithm}

\usepackage[displaymath,textmath,sections,graphics, subfigure, floats]{preview} 
\PreviewEnvironment{center} 
\PreviewEnvironment{pspicture} 

\usepackage{todonotes}
        \def\drawRect#1#2#3#4#5{
           \FPeval{\x2}{(#2) + (#4)} 
           \FPeval{\y2}{(#3) + (#5)} 
           \pspolygon[#1](#2,#3)(\x2,#3)(\x2,\y2)(#2,\y2)
        }

\psset{arrowsize=6pt, labelsep=2pt, linewidth=1.5pt}
\makeatother
\usepackage{hyperref}

\title{Approximating Bin Packing within $O(\log OPT \cdot \log \log OPT)$ bins}

\author{\Large{Thomas Rothvoß}\thanks{Email: {\tt{rothvoss@math.mit.edu}}. Supported by the Alexander von Humboldt Foundation within the Feodor Lynen program, by ONR grant N00014-11-1-0053 and by NSF contract 1115849. 
} \vspace{3mm}  \\
Massachusetts Institute of Technology} 
\date{}

\begin{document}

\maketitle

\begin{abstract}
\noindent 
For \emph{bin packing}, the input consists of $n$ items 
with sizes $s_1,\ldots,s_n \in [0,1]$ which have to be assigned to a minimum
number of bins of size 1. The seminal Karmarkar-Karp algorithm from '82
produces a solution with at most $OPT + O(\log^2 OPT)$ bins.

We provide the first improvement in now 3 decades and show that 
one can find a solution of cost
$OPT + O(\log OPT \cdot \log \log OPT)$ in polynomial time. 
This is achieved by rounding a fractional solution to the
Gilmore-Gomory LP relaxation using the \emph{Entropy Method} from discrepancy theory. 
The result is constructive via algorithms of Bansal and Lovett-Meka.
\end{abstract}

\section{Introduction}

Bin Packing is one of the very classical combinatorial optimization 
problems studied in computer science and operations research. 
It's study dates back at least to the 1950's~\cite{TrimProblem-Eiseman1957} 
and it appeared as one of the prototypical $\mathbf{NP}$-hard problems in 
the book of Garey and Johnson~\cite{GareyJohnson79}. For a detailed account, 
we refer to the survey of \cite{BinPackingSurvey84}.
Bin Packing is also a good case study to demonstrate the development of 
techniques in approximation algorithms. The earliest ones are simple greedy algorithms 
such as the \emph{First Fit} 
algorithm, analyzed by
Johnson~\cite{Johnson73} which requires  at most $1.7\cdot OPT +
1$ bins and \emph{First Fit Decreasing}~\cite{JohnsonFFD74}, which yields a solution with 
at most $\frac{11}{9} OPT + 4$
bins (see~\cite{FFDtightBound-Dosa07} for a tight bound of $\frac{11}{9} OPT + \frac{6}{9}$). 
Later,  Fernandez de la Vega and Luecker~\cite{deLaVegaLueker81} developed 
an \emph{asymptotic PTAS} by introducing an \emph{item grouping technique}
that reduces the number of different item types and has been  
 reused in numerous papers for related problems.
De la Vega and Luecker were able to find a solution of cost at most $(1+\varepsilon)OPT + O(\frac{1}{\varepsilon^2})$ for Bin Packing 
and the running time is either of the form $O(n^{f(\varepsilon)})$ if one uses
dynamic programming or of the form $O(n \cdot f(\varepsilon))$ if one applies 
linear programming techniques.

A big leap forward in approximating bin packing was done by Karmarkar and Karp
in 1982~\cite{KarmarkarKarp82}, who provided an iterative rounding approach 
for the mentioned
linear programming formulation which produces a solution with at most
$OPT + O(\log^2 OPT)$ bins in polynomial time, corresponding to an
\emph{asymptotic FPTAS}.

Both papers \cite{deLaVegaLueker81,KarmarkarKarp82} used the \emph{Gilmore-Gomory LP relaxation} (see e.g. \cite{TrimProblem-Eiseman1957,Gilmore-Gomory61}) 
\begin{equation} \label{eq:GilmoreGomory}
 \min\left\{ {\bf{1}}^Tx \mid Ax = \bm{1}, x \geq \bm{0} \right\}
\end{equation}
where $A$ is the \emph{pattern matrix} that consists of all column vectors $\{ p \in \setZ_{\geq0}^n \mid p^Ts \leq 1\}$.
Each such column $p$ is called a \emph{(valid) pattern} and corresponds
to a feasible multiset of items that can be assigned to a single bin. 
Note that it would be perfectly possible to consider a stronger variant in 
which only patterns $p \in \{ 0,1\}^n$ are admitted. In this case, 
the LP \eqref{eq:GilmoreGomory} could also be interpreted as the
standard {\sc (Unweighted) Set Cover}
relaxation 
\begin{equation} \label{eq:SetCoverLP}
\min\Big\{ \sum_{S \in \mathcal{S}} x_S \mid \sum_{S \in \mathcal{S}: i \in S} x_S \geq 1 \; \forall i \in [n]; x_S \geq 0 \; \forall S \in \mathcal{S}  \Big\}
\end{equation}
for the set system $\mathcal{S} := \{ S \subseteq [n] \mid \sum_{i \in S} s_i \leq 1 \}$.
However, the additive gap between both versions is at most $O(\log n)$ anyway, thus
we stick to the matrix-based formulation as this is more suitable
for our technique\footnote{For example, 
if the input consists of a single item of size $\frac{1}{k}$, then the optimum value
of \eqref{eq:SetCoverLP} is $1$, while the optimum  value of \eqref{eq:GilmoreGomory} is $\frac{1}{k}$. 
But the additive gap can be upper bounded as follows: Take a solution $x$ to \eqref{eq:GilmoreGomory} and  
apply a single grouping via Lemma~\ref{lem:Grouping} with parameter $\beta = 1$. 
This costs $O(\log n)$ and results in a solution to \eqref{eq:SetCoverLP} for some general right hand side 
vector $b \in \setZ_{\geq0}^n$. With the usual cloning argument, this can be easily converted into the form with right hand side $\bm{1}$.}.


Let $OPT$ and $OPT_f$ be the value of
the best integer and fractional solution for \eqref{eq:GilmoreGomory} respectively. 
Although  \eqref{eq:GilmoreGomory} has an exponential number of variables, 
one can compute a basic solution $x$ with $\bm{1}^Tx \leq OPT_f + \delta$ in time polynomial in $n$ and
$1/\delta$~\cite{KarmarkarKarp82} using the Gr{ö}tschel-Lovász-Schrijver
variant of the Ellipsoid method~\cite{GLS-algorithm-Journal81}. 
Alternatively, one can also use the Plotkin-Shmoys-Tardos framework~\cite{FractionalPackingAndCovering-PlotkinShmoysTardos-Journal95} or the
multiplicative weight update method (see e.g. the survey of \cite{MWU-Survey-Arora-HazanKale2012}) to achieve the same guarantee.

The Karmarkar-Karp algorithm operates in $\log n$ iterations in which one 
first groups the items such that only $\frac{1}{2}\sum_{i \in [n]} s_i$ many different item sizes 
remain; then one computes a basic solution $x$ and buys $\lfloor x_p\rfloor$ times 
pattern $p$ and continues with the residual instance. 
The analysis provides a $O(\log^2 OPT)$ upper bound on the 
\emph{additive} integrality gap of \eqref{eq:GilmoreGomory}. 
In fact, it is 
even conjectured in \cite{BinPacking-MIRUP-ScheithauerTerno97} that
\eqref{eq:GilmoreGomory} has the \emph{Modified Integer Roundup Property}, i.e. $OPT \leq \lceil OPT_f\rceil + 1$
(and up to date, there is no known counterexample; the conjecture is known to 
be true for instances that contain at most 7 different
item sizes~\cite{MIRUPproofForDim7-SeboShmonin09}). 
Recently, \cite{BinPackingViaPermutationsSODA2011} found a connection between
coloring permutations and bin packing which shows that 
\emph{Beck's Three Permutations Conjecture} (any 3 permutations
can be bi-colored with constant discrepancy) 
would imply a constant integrality gap
at least for instances with all item sizes bigger than $\frac{1}{4}$.
Note that the gap bound of the Karmarkar-Karp algorithm is actually of the form $O(\log OPT_f \cdot \log( \max_{i,j} \{ \frac{s_i}{s_j} \}))$,
which is $O(\log n)$ for such instances.
But very recently Newman and Nikolov~\cite{CounterexampleBecksPermutationConjecture-FOCS12} found a counterexample to Beck's conjecture.

Considering the gap that still remains between upper and lower bound on the additive integrality gap, 
one might be tempted to try to modify the Karmarkar-Karp algorithm in order
to improve the approximation guarantee. From an abstract point of view, 
\cite{KarmarkarKarp82} buy only patterns that already appear in the initial
basic solution $x$ and then map every item to the slot of a \emph{single}
larger item. 
Unfortunately, combining the insights from \cite{CounterexampleBecksPermutationConjecture-FOCS12} and \cite{BinPackingViaPermutationsSODA2011}, one can show that no
algorithm with this abstract property can yield a $o(\log^2 n)$ gap, which 
establishes a barrier for a fairly large class of algorithms~\cite{BinPackingViaPermutationsTALG2013}.

A simple operation that does not fall into this class is the following: 
\begin{quote}
\emph{Gluing:} Whenever we have a pattern $p$ with $x_p>0$ that has many
copies of the same item, glue these items together and consider them
as a single item.
\end{quote}
In fact, iterating between gluing and grouping, results in a mapping of 
\emph{several} small input items into the slot of a \emph{single} large item -- the
barrier of~\cite{BinPackingViaPermutationsTALG2013} does not hold for such a rounding procedure. 

But the huge problem is: there is no reason why in the worst case, a fractional bin packing solution $x$
should contain patterns with many items of the same type. Also the Karmarkar-Karp
rounding procedure does not seem to benefit from that case either. 
However, there is an alternative algorithm of the author~\cite{EntropyMethodRothvoss-SODA12} to achieve a $O(\log^2 OPT)$ upper bound, 
which is based on \emph{Beck's entropy method}~\cite{Beck-DiscrepancyIntegerSequences-Combinatorica81,Beck-IrregularitiesOfDistributionII-1988,SixStandardDeviationsSuffice-Spencer1985} (or \emph{partial coloring lemma})
from \emph{discrepancy theory}. This is a subfield of combinatorics which deals with the
following type of questions: given a set system $S_1,\ldots,S_n \subseteq [m]$, find a coloring
of the elements $1,\ldots,m$ with red and blue, such that for each set $S_i$
the difference between the red and blue elements (called the \emph{discrepancy}) is as small as possible. 

\section{Outline of the technique}

The partial coloring method is a very flexible technique to color at least half of the elements
in a set system with a small discrepancy, but the technique is based on 
the pigeonhole principle --- with
exponentially many pigeons and pigeonholes --- and is hence non-constructive in 
nature\footnote{The claim of the Partial coloring lemma is as follows: 
Given any vectors $v_1,\ldots,v_n \in \setR^m$ with a parameters $\lambda_1,\ldots,\lambda_n >0$ 
satisfying
\[
  \sum_{i=1}^n G\left( \lambda_i \right) \leq \frac{m}{5}, \quad \textrm{for}\quad G(\lambda) := \begin{cases} 9e^{-\lambda^2/5} & \textrm{if } \lambda \geq 2 \\
\log_2( 32 + \frac{64}{\lambda}) & \textrm{if } \lambda < 2 
\end{cases} 
\]
Then there is a partial coloring $\chi : [m] \to \{ 0,± 1\}$ with  $|\textrm{supp}(\chi)| \geq \frac{m}{2}$ and $|v_i\chi| \leq \lambda_i\|v_i\|_2$ 
for all vectors $i=1,\ldots,n$.}.
But recently Bansal~\cite{DiscrepancyMinimization-Bansal-FOCS2010} and later Lovett and 
Meka~\cite{DiscrepancyMinimization-LovettMekaFOCS12} provided
polynomial time algorithms to find those colorings.
In fact, it turns out that our proofs are even simpler using the 
Lovett-Meka algorithm than using the classical non-constructive version, 
thus we directly use the constructive method. 

\subsection*{The constructive partial coloring lemma}

The Lovett-Meka algorithm provides the following guarantee\footnote{The original statement has $x,y \in [-1,1]^m$ and $|y_j| \geq 1-\delta$ for half of the entries. 
However, one can obtain our version as follows: Start with $x \in [0,1]^m$. 
Then apply \cite{DiscrepancyMinimization-LovettMekaFOCS12} to $x' := 2x-\bm{1} \in [-1,1]^m$ to obtain $y' \in [-1,1]^m$ with $|v_i(y'-x')| \leq \lambda_i\|v_i\|_2$
for all $i \in [n]$ and half of the entries satisfying $|y_j'| \geq 1-\delta$ . Then $y := \frac{1}{2}(y'+\bm{1})$ has half of the entries $y_j \in [0,\frac{\delta}{2}] \cup [1-\frac{\delta}{2},1]$. 
Furthermore, 
$|v_i(y-x)| = \frac{1}{2}|v_i(y'-x')| \leq \frac{\lambda_i}{2}\|v_i\|_2$.}:
\begin{lemma}[Constructive partial coloring lemma~\cite{DiscrepancyMinimization-LovettMekaFOCS12}] \label{lem:ConstructivePartialColoring}
Let $x \in [0,1]^m$ be a starting point, $\delta > 0$ an arbitrary error parameter, $v_1,\ldots,v_n \in \setQ^{m}$ vectors
and $\lambda_1,\ldots,\lambda_n \geq 0$ parameters with 
\begin{equation} \label{eq:EntropyInConstructiveLemma}
\sum_{i=1}^n e^{-\lambda_i^2/16} \leq \frac{m}{16}.
\end{equation}
Then there is a randomized algorithm with expected running time $O(\frac{(m+n)^3}{\delta^2} \log( \frac{nm}{\delta}))$
to compute a vector $y \in [0,1]^m$ with 
\begin{itemize}
\item $y_j \in [0,\delta] \cup [1-\delta,1]$ for at least half of the indices $j \in \{ 1,\ldots,m\}$
\item $|v_iy - v_ix| \leq \lambda_i \cdot \|v_i\|_2$ for each $i \in \{1,\ldots,n\}$.
\end{itemize}
\end{lemma}
If we end up with an almost integral  Bin Packing solution $y$, we can remove 
all entries with $y_j \leq \frac{1}{n}$ and roundup those with $y_j \in [1-\frac{1}{n},1]$ paying only an additional constant term.
Thus we feel free to ignore the $\delta$ term and assume that half of the entries are $y_j \in \{ 0,1\}$.

The algorithm in \cite{DiscrepancyMinimization-LovettMekaFOCS12} is based on a simulated Brownian motion
in the hypercube $[0,1]^m$ starting at $x$. Whenever the Brownian motion hits either the boundary
planes $y_j = 0$ or $y_j=1$ or one of the hyperplanes $v_i(x-y) = ± \lambda_i\|v_i\|_2$,  the Brownian motion
continues the walk in that subspace. By standard concentration bounds, the probability that the 
walk ever hits the $i$th hyperplane is upperbounded by $e^{-\Omega(\lambda_i^2)}$. In other words, 
condition~\eqref{eq:EntropyInConstructiveLemma} says that the expected number of hyperplanes
$v_i(x-y) = ± \lambda_i\|v_i\|_2$ that ever get hit is bounded by $\frac{m}{16}$, from which one can 
argue that a linear number of boundary constraints must get tight. 

Readers that are more familiar with approximation algorithm techniques than with discrepancy theory,
should observe the following: In the special case that $\lambda_i = 0$ for all $i$, one can easily prove 
Lemma~\ref{lem:ConstructivePartialColoring} by choosing $y$ as any basic
solution of $\{ y \mid v_iy=v_ix \; \forall i\in[n]; \; \; \bm{0} \leq y \leq \bm{1}\}$. In other words, Lemma~\ref{lem:ConstructivePartialColoring}
is somewhat an extension of the concept of basic solutions.
Considering that a significant fraction of approximation algorithms is based on the sparse support of 
basic solutions, one should expect many more applications of \cite{DiscrepancyMinimization-LovettMekaFOCS12}.

\subsection*{The rounding procedure}

Let $x$ be a fractional solution for the Gilmore-Gomory LP \eqref{eq:GilmoreGomory}, 
say with $|\textrm{supp}(x)| = m \leq n$ and let $A$ be the constraint matrix
reduced to patterns in the support of $x$. 

Assume for the sake of simplicity that all items have size between $\frac{1}{k}$ and $\frac{2}{k}$ for some $k$.
We now want to discuss how  Lemma~\ref{lem:ConstructivePartialColoring} can be applied in order to
replace $x$ with another vector $y$ that has half of the entries integral and is still almost feasible.
Then repeating this procedure for $\log(m)$ iterations will lead to a completely integral solution.
 For the sake of comparison: 
the Karmarkar-Karp algorithm is able to find another fractional $y$ that has at most half the support of $x$
and is at most an additive $O(1)$ term more costly. So let us argue how to do better. 

Let us sort the items according to their sizes (i.e. $\frac{2}{k} \geq s_1 \geq \ldots \geq s_n \geq \frac{1}{k}$) and 
partition the items into \emph{groups} $I_1,\ldots,I_t$
such that the number of incidences in $A$ is of order $100k$ for each group. In other words, if we 
abbreviate $v_{I_j} := \sum_{i \in I_j} A_i$ as the sum of the row vectors in $I_j$, then  $\|v_{I_j}\|_1 \approx 100k$.
Since each column of $A$ sums up to at most $k$ and each group consums $100k$ incidences, we have only $t \leq \frac{m}{100}$
many groups. Now, we can obtain a suitable $y$ with at most half the fractional entries by either computing a basic solution to the system 
\[
v_{I_j}(x-y) = 0 \;\; \forall j\in[t], \quad \bm{1}^Ty = \bm{1}^Tx,  \quad \bm{0} \leq y \leq \bm{1}
\]
or by applying the Constructive Partial Coloring Lemma to $v_{I_1},\ldots,v_{I_t}$ and $v_{\textrm{obj}}:=(1,\ldots,1)$ 
with a uniform parameter of $\lambda := 0$. In fact, since $(t+1) \cdot e^{-0^2/16} \leq \frac{m}{100}+1 \leq \frac{m}{16}$, 
condition~\eqref{eq:EntropyInConstructiveLemma} is even satisfied with a generous slack.
The meaning of the constraint $v_I(x-y) = 0$ is that $y$ still contains the right
number of slots for items in group $I$. But the constraint does not distinguish between
different items within $I$; so maybe $y$ covers the smaller items in $I$ more often than needed 
and leaves the larger ones uncovered. 
However, it is not hard to argue that after discarding $100k$ items, $y$ can be turned into a feasible solution, 
so the increase in the objective function is again $O(1)$ as for Karmarkar-Karp.

Now we are going to refine our arguments and use the power of the entropy method. 
The intuition is that we want to impose stronger conditions on the coverage of 
items \emph{within} groups.
Consider a group $I := I_j$ and 
create growing \emph{subgroups} $G_1 \subseteq G_2 \subseteq \ldots \subseteq  G_{1/\varepsilon} = I$ such that the number of
incidences grows by $\varepsilon \cdot 100k$ from subgroup to subgroup, for some $\varepsilon > 0$ (later $\varepsilon := \frac{1}{\log^2 n}$
will turn out to be a good choice; see Figure~\ref{fig:RoundingAndGluing}.$(a)$). In other words, $\|v_{G_{j+1} \backslash G_j}\|_1 \approx \varepsilon \cdot 100k$.
We augment the input for Lemma~\ref{lem:ConstructivePartialColoring} by the vectors $v_{G}$ 
for all subgroups equipped with parameter $\lambda_G := 4\sqrt{\ln(\frac{1}{\varepsilon})}$. Observe that
condition~\eqref{eq:EntropyInConstructiveLemma} is still satisfied as each of the $\frac{t}{\varepsilon}$ many subgroups $G$
contributes only $e^{-\lambda_G^2/16} \leq \varepsilon$. So, we can get a better vector $y$ that also satisfies 
$|v_G(x-y)| \leq 4\sqrt{\ln(\frac{1}{\varepsilon})} \cdot \|v_G\|_2$ for any subgroup. In order to improve over our previous approach
we need to argue that $\|v_G\|_2 \ll k$. But we remember that by definition $\|v_G\|_1 \leq 100k$, 
thus we obtain $\|v_G\|_2 \leq \sqrt{\|v_G\|_1 \cdot \|v_{G}\|_{\infty}} \leq 100k \cdot \sqrt{\|v_G\|_{\infty}/\|v_G\|_1}$.
In other words, the only situation in which we do not immediately improve over Karmarkar-Karp 
is if $\|v_G\|_{\infty} \geq \Omega( \|v_G\|_1)$, i.e. if there is some pattern such that a large fraction 
of it is filled with items of the same subgroup $G$.

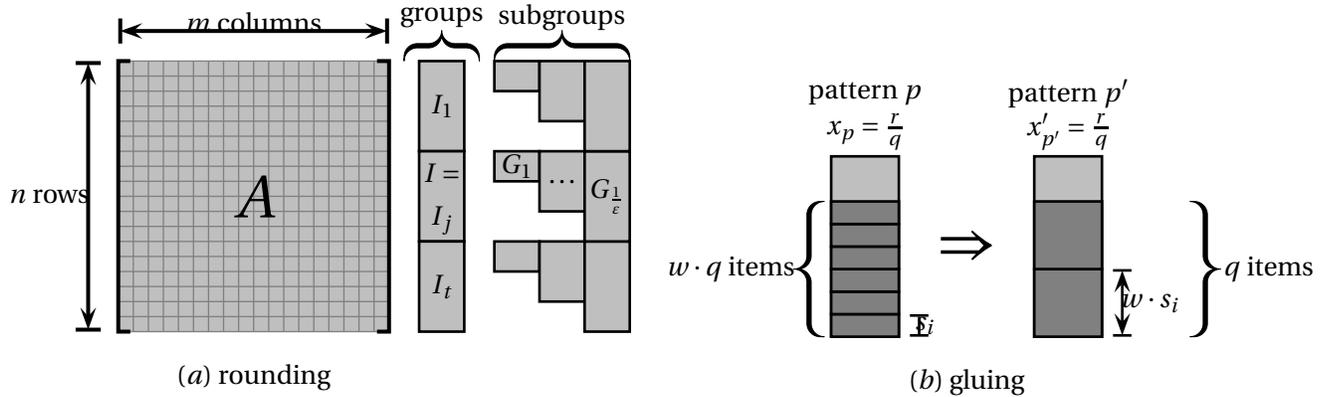
\begin{figure}[H]
\begin{center}
\psset{unit=0.4cm}
\begin{pspicture}(0,-1.3)(10,11)
\drawRect{fillstyle=solid,fillcolor=lightgray,linestyle=none}{0}{0}{9}{9} 
\multido{\N=0.0+0.5}{19}{ 
  \psline[linecolor=gray,linewidth=0.5pt](\N,0)(\N,9)
  \psline[linecolor=gray,linewidth=0.5pt](0,\N)(9,\N)
}
\psline(0.4,0)(0,0)(0,9)(0.4,9) 
\psline(8.6,0)(9,0)(9,9)(8.6,9) 
\rput[c](4.5,4.5){\Huge{$A$}}
\pnode(-1,0){N1} \pnode(-1,9){N2} \ncline{|<->|}{N1}{N2} \naput[labelsep=-2pt]{$n$ rows}
\pnode(0,10){M1} \pnode(9,10){M2} \ncline{|<->|}{M1}{M2} \naput[labelsep=-2pt]{$m$ columns}
\drawRect{linewidth=1pt,fillstyle=solid,fillcolor=lightgray}{10}{6}{1.5}{3} \rput[c](10.75,7.5){$I_1$}
\drawRect{linewidth=1pt,fillstyle=solid,fillcolor=lightgray}{10}{3}{1.5}{3} \rput[c](10.75,5.3){$I=$} \rput[c](10.75,3.7){$I_j$}
\drawRect{linewidth=1pt,fillstyle=solid,fillcolor=lightgray}{10}{0}{1.5}{3} \rput[c](10.75,1.5){$I_t$}
\multido{\N=2+3}{3}{
\drawRect{linewidth=1pt,fillstyle=solid,fillcolor=lightgray}{12.5}{\N}{1.5}{1} 
}
\multido{\N=1+3}{3}{
\drawRect{linewidth=1pt,fillstyle=solid,fillcolor=lightgray}{14}{\N}{1.5}{2} 
}
\multido{\N=0+3}{3}{
\drawRect{linewidth=1pt,fillstyle=solid,fillcolor=lightgray}{15.5}{\N}{1.5}{3} 
}
\rput[c](13.25,5.5){$G_1$}
\rput[c](14.75,5.0){$\ldots$}
\rput[c](16.25,4.5){$G_{\frac{1}{\varepsilon}}$}
\rput[c](4.5,-1.5){$(a)$ rounding}
\psbrace[rot=-90,ref=1C,nodesepB=-5pt,braceWidthInner=5pt,braceWidthOuter=5pt](17,9)(12.5,9){subgroups}
\psbrace[rot=-90,ref=1C,nodesepB=-5pt,braceWidthInner=5pt,braceWidthOuter=5pt](12.0,9)(9.5,9){groups}
\end{pspicture}
\psset{xunit=0.9cm,yunit=0.3cm,linewidth=1pt}
\begin{pspicture}(2,-1.5)(13,10.6)

\drawRect{fillstyle=solid,fillcolor=lightgray}{8}{0}{1}{8}
\drawRect{fillstyle=solid,fillcolor=black!50!white}{8}{0}{1}{1}
\drawRect{fillstyle=solid,fillcolor=black!50!white}{8}{1}{1}{1}
\drawRect{fillstyle=solid,fillcolor=black!50!white}{8}{2}{1}{1}
\drawRect{fillstyle=solid,fillcolor=black!50!white}{8}{3}{1}{1}
\drawRect{fillstyle=solid,fillcolor=black!50!white}{8}{4}{1}{1}
\drawRect{fillstyle=solid,fillcolor=black!50!white}{8}{5}{1}{1}

\drawRect{fillstyle=solid,fillcolor=lightgray}{11}{0}{1}{8}
\drawRect{fillstyle=solid,fillcolor=black!50!white}{11}{0}{1}{3}
\drawRect{fillstyle=solid,fillcolor=black!50!white}{11}{3}{1}{3}
\rput[c](10,4){\Huge{$\Rightarrow$}}
\psbrace[rot=180, ref=1C, nodesepA=-24pt,braceWidthInner=5pt,braceWidthOuter=5pt](7.9,6)(7.9,0){$w\cdot q$ items}
\psbrace[rot=0, ref=1C, nodesepA=18pt,braceWidthInner=5pt,braceWidthOuter=5pt](13.3,0)(13.3,6){$q$ items}
\pnode(9.3,0){A1} \pnode(9.3,1){A2} \ncline{|-|}{A1}{A2} \nbput[labelsep=-3pt]{$s_i$}
\pnode(12.3,0){B1} \pnode(12.3,3){B2} \ncline{|<->|}{B1}{B2} \nbput[labelsep=-2pt]{$w\cdot s_i$}
\rput[c](8.5,10.8){pattern $p$} \rput[c](8.5,9.1){$x_p = \frac{r}{q}$}
\rput[c](11.5,10.8){pattern $p'$} \rput[c](11.5,9.1){$x_{p'}' = \frac{r}{q}$}
\rput[c](10,-2.0){$(b)$ gluing}
\end{pspicture}
\caption{(a) visualization of groups and subgroups. (b) visualization of the gluing procedure. \label{fig:RoundingAndGluing}}
\end{center}
\end{figure}

\subsection*{The gluing}

At this point our gluing operation comes into play. After a simple \emph{pre-rounding step} which costs us a $o(1)$
term, we can assume that all entries in $x$ are multiples of $\frac{1}{q}:=\frac{1}{\log^4 n}$.
Recall that initially we have a single copy
from each item. We group consecutive items together into groups of size 
$\beta = \frac{1}{\log^4 n}$ and round their sizes to the smallest one in the group. 
By standard arguments this incurs a negligible cost of $O(\frac{1}{\log^4 n})$.
Now we can assume that we have a sufficiently large number of copies for every item. 
Suppose that after this agglomeration we find an item $i$ and a pattern $p$ in the support such that
indeed $p_i$ is large, say $p_is_i \geq \frac{1}{\log^8 n}$.
The crucial observation is that this pattern $p$ alone covers $w := \lfloor\frac{p_i}{q}\rfloor $ many copies of item $i$
in the input since $x_p\cdot p_i \geq w$.
Next, take $w$ many copies of item $i$ in $p$ and glue them together to obtain
a new, bigger item $i'$ of size $s_{i'} = w\cdot s_i$. The pattern $p$ has enough items
to do this $q$ times, see Figure~\ref{fig:RoundingAndGluing}.$(b)$.
In other words, the modified pattern now contains  $q$ copies of a new artificial item $i'$. 
The reason why we want $q$ copies of this new item is that the modified pattern $p$
alone covers $q \cdot x_p = 1$ copies of $i'$. Thus, in a finally obtained
integral solution we would have a slot for the artificial item $i'$, which we can then 
replace with the copies of the original item $i$. 

Observe that the size of this newly obtained item type is $s_{i'} = w \cdot s_i \geq \frac{1}{\log^{12} n}$.
So we call items above that size \emph{large} and below that size \emph{small}.
The interesting effect is that if we apply this gluing procedure to all small items
whenever possible, the penalty that we pay for rounding the remaining small items
is so small that the overall cost is completely dominated by the contribution of the large items
(i.e. those items that were either large from the beginning or that were created during the gluing
process). 
In other words, we obtain the same approximation
guarantee as if the instance would only contain items of size at least $\frac{1}{\log^{12} n}$ from the 
beginning on; for those instances already \cite{KarmarkarKarp82} produces a
solution with at most  $OPT_f + O(\log n \cdot \log \log n)$ bins, so this is our 
final approximation guarantee for all instances.



\subsection*{Contribution}

Our main contribution is the following theorem:
\begin{theorem} \label{thm:MainContribution}
For any Bin Packing instance $s_1,\ldots,s_n \in [0,1]$, one can compute 
a solution with at most $OPT_f + O(\log OPT_f \cdot \log \log OPT_f)$ bins in expected
time $O(n^6 \log^5(n))$, where $OPT_f$ denotes the optimum value of the Gilmore-Gomory LP relaxation.
\end{theorem}
This partly solves problem $\#3$ in the list of 10 open problems in approximation algorithms stated
by Williamson and Shmoys~\cite{DesignOfApproxAlgosWilliamsonShmoys} (they asked for a constant integrality gap).

\section{Related work}

The classical application of the partial coloring lemma is to
find a coloring $\chi : [m] \to \{ ± 1\}$ for $m$ elements such that for a given set system $S_1,\ldots,S_n$\footnote{The standard notation in 
discrepancy theory is to have $n$ as number of elements and $m$ as the number of sets. 
However, that conflicts with the standard notation for Bin Packing, where $n$ is the number of items
which is essentially the number of $v$-vectors.} 
the \emph{discrepancy} $\max_{i \in [n]} |\sum_{j \in S_i} \chi(j)|$ is minimized. For example, one can obtain Spencer's bound~\cite{SixStandardDeviationsSuffice-Spencer1985}
on the discrepancy of arbitrary set systems, by
applying $\log m$ times Lemma~\ref{lem:ConstructivePartialColoring} starting with $x := (\frac{1}{2},\ldots,\frac{1}{2})$ and a 
uniform bound of  $\lambda := C\sqrt{\log \frac{2n}{m}}$ where $v_i \in \{ 0,1\}^{m}$ is the characteristic vector of $S_i$.
This results in a coloring $\chi :[m] \to \{ ± 1\}$ with $|\chi(S)| \leq O(\sqrt{m\log \frac{2n}{m}})$. 
Note that e.g. for $n \leq O(m)$, this is a $O(\sqrt{m})$ coloring, while a pure random coloring
would be no better than $O(\sqrt{m \cdot \log m})$.

Other applications of this method give
a $O(\sqrt{t} \log m)$ bound if no element is in more than $t$ sets~\cite{DiscrepancyBound-sqrtT-logN-SrinivasanSODA97} and a $O(\sqrt{k}\log m)$ bound for the discrepancy of $k$ permutations~\cite{DiscrepancyOfPermutations-SpencerEtAl}. For the first quantity, alternative proof techniques
give bounds of $2t-1$~\cite{IntegerMakingTheorems-BeckFiala81} and $O(\sqrt{t \cdot \log m})$~\cite{BalancingVectors-Banaszczyk98}.

In fact, we could use those classical techniques and extend \cite{EntropyMethodRothvoss-SODA12}
to obtain a $OPT_f + O(\log OPT_f \cdot \log \log OPT_f)$ \emph{integrality gap} result.
It might appear surprising that one can bound integrality gaps by coloring
matrices, but this is actually a well known fact, which is 
expressed by the Lovász-Spencer-Vesztergombi Theorem~\cite{DiscrepancyofSetSystemsAndMatrices-LSV86}: 
Given a matrix $A$ and a vector $x \in [0,1]^m$ such that any submatrix of $A$ admits a 
discrepancy $\alpha$ coloring. Then there is a $y \in \{ 0,1\}^m$ with $\|Ax - Ay\|_{\infty} \leq \alpha$.
For a more detailed account on discrepancy theory, we 
recommend Chapter~4 in the book of Matou{\v{s}}ek~\cite{GeometricDiscrepancy-Matousek99}.


\section{Preliminaries}

In the Bin Packing literature, it is well known that
it suffices to show  bounds as in Theorem~\ref{thm:MainContribution}
with an $n$ instead of $OPT_f$ and that one can also assume that items are not too tiny,
e.g. $s_i \geq \frac{1}{n}$. 
Though the following 
arguments are quite standard (see e.g. \cite{KarmarkarKarp82}), we present them for the sake
of completeness.
\begin{lemma} \label{lem:AssumptionSizesAtLeast1-n}
Assume for a monotone function $f$, there is a $\textrm{poly}(m)$-time $OPT_f + f(n)$ algorithm for Bin Packing instances $s \in [0,1]^m$ with $|\{ s_i \mid i \in [m]\}| \leq n$
many different item sizes and $\min\{ s_i \mid i \in [m]\} \geq \frac{1}{n}$. Then there is a polynomial time algorithm that finds 
a solution with at most 
$OPT_f + f(OPT_f) + O(\log OPT_f)$ bins.
\end{lemma}
\begin{proof}
Let $s \in [0,1]^m$ be any bin packing instance and define $\sigma := \sum_{i=1}^m s_i$ as their size.
First, split items into large ones $L := \{ i \in [m] \mid s_i \geq \frac{1}{\sigma}\}$
and small ones $S := \{ i \in [m] \mid s_i < \frac{1}{\sigma}\}$.

We perform the grouping procedure from \cite{KarmarkarKarp82} (or from Lemma~\ref{lem:Grouping}) 
to large items $L$ 
and produce an instance with sizes $s'$ such that each size $s_i'$ that appears 
has $\sum_{j: s_j'=s_i'} s_i' \geq 1$. Moreover, after discarding items of total size at 
most $O(\log \frac{1}{\min\{ s_i \mid i \in L\}}) \leq O(\log \sigma)$
one has $OPT_f' \leq OPT_f$. 
Thus the number of different item sizes in $s'$ is bounded by $\sigma$. 
We run the assumed algorithm to assign items in $L$ to at most $OPT_f' + f(\sigma) \leq OPT_f + f(OPT_f)$
bins (using that $OPT_f \geq \sigma$ and $f$ is monotone). Adding the discarded items increases the objective function
by at most another $O(\log OPT_f)$ term.
Now we assign the small items greedily over those bins. 
If no new bin needs to be opened, we are done. 
Otherwise, we know that the solution consists of $k$ bins such that  $k-1$ bins are at least $1 - \frac{1}{\sigma}$ full. This implies 
$\sigma \geq (k-1) \cdot (1 - \frac{1}{\sigma})$, and hence $k \leq \sigma + 3 \leq OPT_f + 3$ assuming $\sigma \geq 2$.
\end{proof}
From now on, we have the implicit assumption $s_i \geq \frac{1}{n}$. 
In an alternative Bin Packing definition, also called the \emph{cutting stock problem}, the input 
consists of a pair $(s,b)$ 
such that $b_i \in \setZ_{\geq0}$ gives the number of copies of $s_i$. 
The Karmarkar Karp bound of $O(\log^2 n)$ on the additive integrality gap
still holds true in this general setting, where $n$ is the number of item types. 
Note that the time to solve the LP (\ref{eq:GilmoreGomory}) up to an additive constant
is polynomial in $\sum_{i=1}^n b_i$. 
In this paper, we will work with a more general formulation in which 
any $b \in \textrm{cone}\{ p \in \setZ_{\geq0}^n \mid s^Tp \leq 1\}$ may serve as vector of multiplicities
(note that such a vector might have fractional entries).
From our starting solution $x$, we can immediately remove the integral parts $\lfloor x_p\rfloor$
and assume that $\bm{0} \leq x < \bm{1}$, which has the consequence that $\sum_{i=1}^n s_ib_i < n$.


It will be useful to reformulate bin packing as follows: 
consider a size vector $s \in [0,1]^n$ ($s_1 \geq \ldots \geq s_n$) with pattern matrix $A$ 
and a given vector $x \in \setR_{\geq0}^m$ as input 
and aim to solve the following problem
\begin{eqnarray}
 \min \bm{1}^Ty & & \quad \label{eq:BinPackingReformulation} \\
\sum_{j\leq i} A_jy &\geq& \sum_{j\leq i} A_jx \quad \forall i \in [n] \nonumber \\
y &\in& \setZ_{\geq0}^{m} \nonumber
\end{eqnarray}
We write $y \succeq x$ if $\sum_{j\leq i} A_jy \geq \sum_{j\leq i} A_jx$ for all $i \in [n]$. 
In words: we have a fractional solution $x$ to LP~\eqref{eq:GilmoreGomory} for an instance with $A_ix$
many items of type $i$ in the input and aim to find an integral solution $y$
that reserves $\sum_{j\leq i} A_jy$ many slots for items of type $1,\ldots,i$.
The condition $y \succeq x$ guarantees that $y$ can be easily transformed into 
a feasible solution by simply assigning items to slots of larger items.
We make the following observation:
\begin{observation}
Consider any instance
 $1\geq s_1 \geq \ldots \geq s_n > 0$ with pattern matrix $A$ and vector $x \in \setR_{\geq0}^m$ such that $Ax = \bm{1}$. 
Then the value of the optimum integral solution to \eqref{eq:BinPackingReformulation} and \eqref{eq:GilmoreGomory} coincide.
\end{observation}
However,  \eqref{eq:BinPackingReformulation} has the advantage that we can split the solution $x = x' + x''$   
and then separately consider $x'$ and $x''$ while the vector $b' = Ax'$ might be fractional,
which is somewhat unintuitive when speaking about classical bin packing.
When $Ax \in \setZ_{\geq0}^n$ and $y$ with $y \succeq x$ is integral, then it is clear that
$y$ defines a solution in which each item represented by multiplicity vector 
$Ax$ can be mapped to one slot in the patterns of $y$. 

\subsection*{Notation}

To fix some notation, $p$ denotes a pattern which we interpret either
as a multi-set of items or
as a vector where $p_i \in \setZ_{\geq 0}$ 
denotes the number of copies of item $i$ contained in $p$. 
The matrix formed by all possible patterns is denoted by $A$. 
Moreover $A_i$ is the $i$th row of $A$ and
by a slight abuse of notation, sometimes we interpret $p$ as a column index 
for pattern $p$ and write $A^p$ as the $p$th column. 
As usual $[n] = \{ 1,\ldots,n\}$ and $\bm{1}$ denotes the all-ones vector of suitable
dimension. For a subset $I \subseteq [n]$, we write $s(I) := \sum_{i \in I} s_i$.
For any $k$ that is a power of 2, we denote the subset of items $\{ i \in [n] \mid \frac{1}{k} \leq s_i < \frac{2}{k}\}$ as one \emph{size class}.
The quantity $m$ will usually refer to the number of patterns in $A$. 

%
\section{Operations on fractional solutions}

We introduce two useful operations that we can apply to a fractional solution: 
the classical \emph{item grouping procedure} similar to \cite{deLaVegaLueker81,KarmarkarKarp82}
and a novel \emph{item gluing} operation. Finally, we show how they can be combined to
obtain a \emph{well spread} instance in which no pattern contains a significant fraction
of copies of a single item.

In order to keep the maintained solution feasible in these procedures 
it will be necessary to add some additional patterns.
In the classical literature~\cite{deLaVegaLueker81,KarmarkarKarp82} 
this would be done with the phrase
``discard the following set of items\ldots'' meaning that those items
are assigned to separate bins in a greedy manner. 
We choose to handle this slightly differently. 
We allow additional columns
in $A$ -- for each $i \in [n]$, we add a \emph{waste pattern} $\{i\}$, which can be
bought in arbitrary fractional quantities at cost $2s_i$ per copy. 
For a vector $x \in \setR_{\geq0}^m$ representing a fractional solution,
we write
\[
  (\bm{1},2s)^Tx = \sum_{p\textrm{ regular pattern}} x_p + \sum_{\{i\}\textrm{ waste pattern}} 2s_ix_{\{i\}}
\]
as objective function. During our rounding algorithm, we do not 
make any attempt to round entries $x_{\{i\}}$ belonging to 
waste patterns to integral values. This can be easily done
at the very end as follows:
\begin{lemma} \label{lem:RoundingFractionalWastePatterns}
Let $x \in \setR^m_{\geq0}$ and suppose that all patterns  $p \in \supp(x)$ contain 
only one item, i.e. $\| p \|_1 = 1$. Then there is an integral $y \succeq x$ with $\sum_{i=1}^n s_iy_{\{i\}} \leq \sum_{i=1}^{n} s_ix_{\{i\}} + 1$.
\end{lemma}
\begin{proof}
By adding dummy copies, we may assume that $x_p = \frac{1}{q}$ for all $p$
(for some large number $q$). Sort the patterns $p_1,\ldots,p_{Q}$ such that the item sizes in those
patterns are non-increasing. Buy each $q$th pattern starting with $p_q$ plus 
one copy of $p_1$.
\end{proof}
Finally, any set of Bin Packing items $S \subseteq [n]$ 
can be assigned to at most $2 \sum_{i \in S} s_i + 1$ bins
using a First Fit assignment, which is the reason for the
penalty factor of $2$ for waste patterns.


\subsection{Grouping}

The operation of \emph{grouping items} is already defined by de la Vega and Luecker in 
their asymptotic PTAS for Bin Packing~\cite{deLaVegaLueker81}. 
For some parameter $k$, they form groups of $k$ input items each and round up 
the item sizes to the size of the largest item in that group. This essentially reduces
the number of different item types by a factor of $k$. In contrast, we will 
replace items in the \emph{fractional solution} $x$ with \emph{smaller} items.
The reason for our different approach is that we measure progress in our algorithm
in terms of $|\supp(x)|$, while e.g. Karmarkar-Karp measure the progress in
terms of the total size of remaining input items. As a consequence we have to be careful that
no operation increases $|\supp(x)|$.
\begin{lemma}[Grouping Lemma] \label{lem:Grouping}
Let $x \in \setR_{\geq0}^m$ be a vector, $\beta > 0$ any parameter   
and $S \subseteq [n]$ be a subset of items.
Then there is an $x' \succeq x$ with identical fractionality as $x$ (except of waste patterns) 
with  $(\bm{1},2s)^Tx' \leq (\bm{1},2s)^Tx + O(\beta \cdot \log (2\max_{i,j \in S} \{ \frac{s_i}{s_j} \})$, and for any $i \in S$, 
either $A_ix' = 0$ or $s_iA_ix' \geq \beta$.
\end{lemma}

\begin{proof}
It suffices to consider the case in which $\alpha \leq s_i \leq 2\alpha$ for all $i \in S$ and
show that the increase in the 
objective function is bounded by $O(\beta)$. 
The general
case follows by applying the lemma to all size classes
$S \cap \{ i \in [n] \mid (\frac{1}{2})^{\ell+1} < s_i \leq (\frac{1}{2})^{\ell}\}$.
We also remove those items that have already $s_iA_ix\geq\beta$ from $S$
since there is nothing to do for them.

In the following, we assume that items are sorted according to their sizes.
We consider the index set $I := \{(i,p) \mid i \in S, p \in \supp(x)\}$. 
For any subset $G \subseteq I$, we define the weight as $w(G) := \sum_{(i,p) \in G} s_ip_ix_p$.
Note that any single index has weight $w(\{(i,p)\}) = s_ip_ix_p \leq s_iA_ix \leq \beta$ by assumption.
Hence we can partition $I = G_1 \dot{\cup} \ldots \dot{\cup} G_r$ such that 
\begin{itemize}
\item $w(G_{k}) \in [2\beta,4\beta] \; \forall k=1,\ldots,r-1$
\item $w(G_r) \leq 2\beta$
\item $(i,p) \in G_k, (i',p') \in G_{k+1} \Rightarrow i\leq i'$
\end{itemize}
Now, for each $k \in \{ 1,\ldots,r-1\}$ and each index $(i,p) \in G_k$, we replace items of
type $i$ in $p$ with the \emph{smallest} item type that appears in $G_k$.
Furthermore, for indices $(i,p) \in G_r$, we remove items of type $i$ from $p$.
Finally, we add $4\frac{\beta}{\alpha}$ many copies of the largest item in $I$ to the waste 
(note that the number $4\frac{\beta}{\alpha}$ can be fractional and even $4\frac{\beta}{\alpha}\ll1$ 
is meaningful). Let $x'$ denote the emerging solution. Clearly, $x'$ only uses 
patterns that have size at most $1$. Moreover, $(\bm{1}^T,2s)x' - (\bm{1}^T,2s)x \leq 4\frac{\beta}{\alpha} \cdot 2\max\{ s_i \mid i \in S\} \leq 16\beta$. 

It remains to argue that $x' \succeq x$. Consider any item $i \in [n]$ and the
difference $\sum_{j\leq i} A_jx' - \sum_{j\leq i} A_jx$. There is at most one group $G_{k}$ whose
items were (partly) larger than $i$ in $x$ and then smaller in $x'$. 
The weight of that group is $w(G_{k}) \leq 4\beta$, thus their ``number'' is
$\sum_{(i,p) \in G_{k}} p_ix_p \leq \frac{4\beta}{\alpha}$. We add at least this ``number'' of items 
to the waste, thus
\[
  \sum_{j\leq i} A_jx' - \sum_{j\leq i} A_jx \geq \frac{4\beta}{\alpha} - \sum_{(i,p) \in G_k} p_ix_p \geq 0
\]
\end{proof}

\subsection{Gluing}

We now formally introduce our novel item gluing method. 
Assume we would a priori know some set of items which in 
an optimal integral solution is assigned to the same bin. 
Then there would be no harm in gluing these items together
to make sure they will end up in the same bin. 
The crucial point is that this is still possible with copies of an item $i$
appearing in the same pattern in a \emph{fractional} solution
as long as the contribution $x_p\cdot p_i$ to the multiplicity vector is integral, 
see again Figure~\ref{fig:RoundingAndGluing}.$(b)$.


\begin{lemma}[Gluing Lemma] \label{lem:GluingLemma}
Suppose that there is a pair of pattern $p$ and item $i$ with $x_p = \frac{r}{q}$  and $p_{i} \geq w\cdot q$ ($r,q,w \in \setN$)
as well as a size $s_{i'} = w \cdot s_i$.
Modify $x$ such that $w\cdot q$ items of type $i$ in pattern $p$ are replaced
by $q$ items of type $i'$ and call the emerging solution $x'$.
Then the following holds:
\begin{enumerate}
\item[a)] The patterns in $x'$ have still size at most one and ${\bm{1}}^Tx={\bm{1}}^Tx'$.
\item[b)] Any integral solution $y'\succeq x'$ can be transformed into an integral solution $y \succeq x$
of the same cost.
\end{enumerate}
\end{lemma}
\begin{proof}
The first claim is clear as $q\cdot s_{i'} = wq\cdot s_i$.


Now, let  $y'\succeq x'$ be an integral solution. Recall that $A_{i'}x' \geq \frac{r}{q}\cdot q = r \in \setZ_{>0}$. 
Select the $r$ smallest slots of size at least $s_{i'}$ that are contained in $Ay'$. Substitute each such slot
with $w$ items of type $i$ and call the emerging solution $y$. Note that $y$ is integral with $y \succeq x$.
\end{proof}

Any sequence of grouping and gluing produces a solution which dominates
the original instance in the sense that any integral solution for the
transformed instance implies an integral solution for the original one.
\begin{corollary} \label{cor:SequenceRoundingAndGluing}
Let $s \in [0,1]^m$ and $x \in \setR_{\geq0}^m$ be any instance for (\ref{eq:BinPackingReformulation}).
Suppose there is a sequence $x= x^{(0)},\ldots,x^{(T)}$ with $x^{(T)} \in \setZ_{\geq0}^m$ such that for each $t \in \{1,\ldots,T\}$, 
at least one of the cases is true: 
\begin{itemize}
\item (i) $x^{(t)} \succeq x^{(t-1)}$ 
\item (ii) $x^{(t)}$ emerges from $x^{(t-1)}$ by gluing items via Lemma~\ref{lem:GluingLemma}. 
\end{itemize}
Then one can construct an integral solution $y \in \setZ_{\geq0}^m$ with $y \succeq x$ and $\bm{1}^Ty \leq \bm{1}^Tx^{(T)}$
in polynomial time.
\end{corollary}
\begin{proof}
Follows by induction over $T$, the definition of ``$\succeq$'' and Lemma~\ref{lem:GluingLemma}.b).
\end{proof}

\subsection{Obtaining a well-spread instance}

As already argued in the introduction, a rounding procedure based on the 
partial coloring method would beat \cite{KarmarkarKarp82} if the patterns
in $p \in \supp(x)$ would satisfy $A_{ip} \leq \delta \cdot A_ix$ for $\delta \leq o(1)$. 
We call this property \emph{$\delta$-well-spread w.r.t. $\varepsilon$-small items}.
A crucial lemma is to show that we can combine grouping and gluing to
obtain a $\frac{1}{\polylog(n)}$-well-spread solution 
while loosing a negligible additive $\frac{1}{\polylog(n)}$ term
in the objective function. 

To simplify notation, let us assume that  the vector $s$ contains
already all sizes $k\cdot s_i$ for $i=1,\ldots,n$ and $k \in \setN$ (even if $x$ does not
contain any item of that size). 
\begin{lemma} \label{lem:WellSpreadInstance}
Let $1\geq s_1 \geq \ldots \geq s_n \geq \frac{1}{n}$ and $x \in [0,1]^m$ be given 
such that for some  $q \in \setZ_{>0}$ one has $x_p \in \frac{\setZ_{\geq0}}{q}$ for all $p \in \supp(x)$ and
$|\supp(x)| \leq n$. Choose any parameters $\delta, \beta > 0$ and call items of size at least $\varepsilon := \delta\frac{\beta}{2q}$ \emph{large} and \emph{small}
otherwise.  
Then one can apply Grouping and Gluing to obtain a solution $\tilde{x}$
with $(\bm{1},2s)^T\tilde{x} \leq (\bm{1},2s)^Tx + O(\beta \log^2 n)$ and the property that 
$p_i \leq \delta\cdot A_i\tilde{x}$ for all small items $i$ and all $p \in \supp(\tilde{x})$.
\end{lemma}
\begin{proof}
First apply grouping with parameter $\beta$ to the small items to obtain 
a vector $x' \succeq x$ with $(\bm{1},2s)^Tx' \leq (\bm{1},2s)^Tx + O(\beta \cdot \log n)$
such that $s_i\cdot A_ix' \geq \beta$ whenever $A_ix'>0$.
Now apply gluing for each $i$ and $p \in \supp(x')$, 
wherever $s_iA_{ip} \geq 2\delta\cdot\beta$ with maximal possible $w$. 
In fact, that means $w \geq \lfloor\frac{2\delta \beta}{qs_i}\rfloor \geq \frac{\delta\beta}{qs_i}$ since
$\frac{\delta\beta}{qs_i} \geq \frac{\delta \beta}{q\varepsilon} \geq 1$.
The size of the items emerging from the
gluing process is at least $w\cdot s_i \geq \frac{\delta\beta}{q}$, thus they are large by definition.
We have at most $q$ items of type $i$ remaining in the pattern and 
their total size is $q \cdot s_i \leq q\cdot \varepsilon \leq \frac{\delta\beta}{2}$.
Let $x''$ be the new solution.

If after gluing, we still have $s_iA_ix'' \geq \frac{\beta}{2}$, then we say $i$ is \emph{well-covered}.
If indeed all small items are well-covered, then we are done 
because $s_iA_{ip} \leq s_iq \leq \frac{\delta\beta}{2} \leq \delta \cdot s_iA_ix''$ for all small $i$ and $p \in \supp(x'')$.

Thus, let $S := \{ i \textrm{ small} \mid i\textrm{ not well-covered}\}$ be the set of those items whose number has
decreased to less than half due to gluing.
We apply again grouping (Lemma~\ref{lem:Grouping}) to $S$ (note that we do not 
touch the well-covered items).
Then we apply again gluing where ever possible and repeat the procedure until 
all small items are well-covered. Note that once an item is well-covered it stays well-covered
as it is neither affected by grouping nor by gluing.

In each iteration the waste increases by $O(\beta \cdot \log n)$, thus it suffices to argue that
the procedure stops after at most $O(\log n)$ iterations. 
Note that the total size of not well-covered items  $\sum_{i\textrm{ not well-covered}}s_iA_ix$ decreases
by at least a factor of $\frac{1}{2}$
in each iteration. Moreover at the beginning we had $\sum_{i=1}^n s_iA_ix < n$ 
and we can stop the procedure when  $\sum_{i\textrm{ not well-covered}}s_iA_ix \leq \frac{1}{n^2}$\footnote{In fact, whenever we have a pattern $p$ with $x_p \leq \frac{1}{n}$ we can just move it to the waste. In total over all iterations this does not cost us more than an extra $1$ term. 
Then we always have the trivial lower bound $s_iA_ix \geq \frac{1}{n^2}$ as $s_i \geq \frac{1}{n}$.}, which
shows the claim. 
\end{proof}

\section{The algorithm}

In this section, we present the actual rounding algorithm, which can be informally
stated as follows (we give a more formal definition later):
\begin{enumerate}
\item[(1)] FOR $\log n$ iterations DO
  \begin{enumerate}
  \item[(2)] round $x$ s.t. $x_p \in \frac{\setZ_{\geq0}}{\polylog(n)}$ for all $p$
  \item[(3)] make $x$ $\frac{1}{\polylog(n)}$-well spread
  \item[(4)] run the constructive partial coloring lemma to make half of the variables integral
  \end{enumerate}
\end{enumerate}
 
For the sake of comparison note that the Karmarkar-Karp algorithm~\cite{KarmarkarKarp82}
consists of step (1) + (4), just that the application
of the constructive partial coloring lemma is replaced with 
grouping + computing a basic solution.

\subsection{Finding a partial coloring}

The next step is to show how Lemma~\ref{lem:ConstructivePartialColoring} 
can be applied to make at least half of the variables integral.
As this is the crucial core procedure in our algorithm, we present it as a stand-alone
theorem and list all the properties that we need for matrix $A$.
Later, we will apply Theorem~\ref{thm:ConstructivePartialColoringForBinPacking} to 
the matrix of patterns $p$ that have $0<x_p<1$, after making the solution well-spread.
Mathematically speaking, the point is that any matrix $A$ that is column-sparse and has 
well-spread rows admits good colorings via the entropy method.

\begin{theorem} \label{thm:ConstructivePartialColoringForBinPacking}
Let $x \in [0,1]^m$ be a vector and $\delta,\varepsilon$ be parameters with $0<\varepsilon\leq\delta^2 \leq 1$ and let  $A \in \setZ_{\geq0}^{n × m}$ (with  $m\geq100 \log (\max_i \{ \frac{2}{s_i}\})$) be any matrix with numbers $1  \geq s_1 \geq \ldots \geq s_n > 0$. Suppose that for any column $p \in [m]$, one has
 $A^p s = \sum_{i=1}^n A_{ip}s_i \leq 1$ and for any row $i$ with $s_i \leq \varepsilon$ one has $\|A_i\|_{\infty} \leq \delta \cdot \|A_i\|_1$. 
Then there is a randomized algorithm with expected polynomial running time
to compute a $y \in [0,1]^m$ with $|\{ p \in [m] \mid y_p \in \{ 0,1\}\}| \geq \frac{m}{2}$, $\bm{1}^Ty = \bm{1}^Tx$ and for all $i \in [n]$
\[
  \Big|\sum_{j\leq i} A_{j}(y-x)\Big| 
\leq \begin{cases}
  O(\frac{1}{s_i}) & s_i > \varepsilon \\
  O(\sqrt{\ln(\frac{2}{\delta})\cdot\delta} \cdot \frac{1}{s_i}) & s_i \leq \varepsilon 
  \end{cases}
\]
\end{theorem}
\begin{proof}
First of all, it will be convenient for our arguments if each individual row
has a small norm. Consider a row $A_i$ belonging to a large index  (i.e. $s_i > \varepsilon$)
and replace it with $\|A_i\|_1$ many rows that sum up to $A_i$, each having unit $\|\cdot\|_1$-norm.
Note that any $y$ that satisfies the claim for the modified matrix does so for the 
original matrix $A$. 
Similarly, consider any index $i$ belonging to a small item (i.e. $s_i \leq \varepsilon$). 
Then we know that $\|A_i\|_{\infty} \leq \delta \cdot \|A_i\|_1$. Thus we can replace $A_i$
by non-negative integral row vectors that sum up to $A_i$, each of which has $\| \cdot \|_{\infty}$-norm 1 and $\| \cdot \|_{1}$-norm 
between $\frac{1}{2\delta}$ and $\frac{1}{\delta}$.  After this replacement we can assume that each index $i$ with $s_i \leq \varepsilon$
satisfies  $\frac{1}{2\delta} \leq \|A_i\|_1 \leq \frac{1}{\delta}$ and $\|A_i\|_{\infty} = 1$. 

In the following let $C>0$ be a large enough constant that we determine later.
We will prove the claim via a single application of Lemma~\ref{lem:ConstructivePartialColoring}.
In particular we need to make a choice of vectors $v_i$ and parameters $\lambda_i$.
First, we partition the items into \emph{groups} $[n] = I_1 \dot{\cup} \ldots \dot{\cup} I_t$ such that 
each group $I$ consists of consecutive items and is chosen maximally such that $\sum_{i\in I} \|A_i\|_1 s_i \leq C$
and so that $I$ contains only items from one size class. In other words, apart from the $\log (\max_i\{ \frac{2}{s_i}\})$
many groups that contain the last items from some size class, we will have that $\sum_{i \in I} \|A_i\|_1s_i \geq C-1$.
For each group $I$, whether complete or not, we define a vector $v_I := \sum_{i \in I} A_i$ with parameter $\lambda_I := 0$.

Now consider a group $I$ that belongs to small items, say the items $i \in I$ have size $\frac{1}{k} \leq s_i \leq \frac{2}{k}$ for some $k$.
We form \emph{subgroups} $G_1 \subseteq G_2 \subseteq \ldots \subseteq G_{t(I)} = I$ 
such that $\|v_{G_{j+1} \backslash G_j}\|_1 \in [C\sqrt{\delta}k,2C\sqrt{\delta}k]$ (this works out since $\|A_i\|_1 \leq \frac{1}{\delta} = \frac{\sqrt{\delta}}{\delta^{3/2}} \leq \frac{\sqrt{\delta}}{\varepsilon} \leq \sqrt{\delta}k$; again the last subgroup $G_{t(I)}$ might be smaller). 
For each subgroup $G$, we add the vector $v_G$ to our list, equipped with error parameter $\lambda_G := 4\sqrt{\ln( \frac{2}{\delta})}$.

To control the objective function, we also add the all-ones vector $v_{\textrm{obj}} := (1,\ldots,1)$ with $\lambda_{\textrm{obj}} := 0$.
Now we want to argue that Lemma~\ref{lem:ConstructivePartialColoring}, applied to all the vectors $v_I,v_G,v_{\textrm{obj}}$ defined above, provides a solution $y$ that satisfies the claim. 
The first step is to verify that indeed the ``entropy condition'' \eqref{eq:EntropyInConstructiveLemma} is satisfied.
As the item sizes for each column sum up to most one, we know that $\sum_{i \in [n]} \|A_i\|_1s_i \leq m$, 
where $m$ is the number of columns. Each complete group
has $\sum_{i \in I} \|A_i\|_1s_i \geq C-1$, thus the number of groups is $t \leq \frac{m}{C-1} + \log( \max_i\{ \frac{2}{s_i}\}) \leq \frac{m}{50}$ for $C$ large enough
and each group $I$ contributes $e^{-\lambda_I^2/16} = 1$ to \eqref{eq:EntropyInConstructiveLemma}.

Next, consider any group $I$ and let us calculate the contribution just of its subgroups $G_1,\ldots,G_{t(I)}$
to \eqref{eq:EntropyInConstructiveLemma}. The number of $I$'s subgroups is  $t(I) \leq \frac{1}{\sqrt{\delta}} + 1 \leq \frac{2}{\sqrt{\delta}}$
 and each subgroup $G$ contributes $e^{-\lambda_G^2/16} = \frac{\delta}{2}$, thus their total contribution 
is bounded by $\frac{2}{\sqrt{\delta}} \cdot \frac{\delta}{2} \leq 1$. In other words, the total contribution of all subgroups
is bounded from above by $\frac{m}{100}$ as well and \eqref{eq:EntropyInConstructiveLemma} indeed holds
and we can apply Lemma~\ref{lem:ConstructivePartialColoring}. The algorithm returns a vector $y$ such that
$|v_I(x-y)| = 0$ for each group $I$ and $|v_G(x-y)| \leq 4\sqrt{\ln( \frac{2}{\delta})}\cdot\|v_G\|_2$ for each subgroup $G$
(and of course $\bm{1}^Tx = \bm{1}^Ty$).

Finally, consider any item $i$ and suppose it is small with $\frac{1}{k} \leq s_i \leq \frac{2}{k}$. 
It remains to show that
$|\sum_{j\leq i} A_j(x-y)| \leq O(\sqrt{\delta\ln(\frac{2}{\delta})} \cdot k)$.
Now we use that the interval $\{1,\ldots,i\}$
can be written as disjoint union of a couple of groups + a single subgroup + a small rest.
So let $t'$ be the index with $i \in I_{t'}$.
Moreover, let $G$ be the (unique) maximal subgroup such that $G \subseteq \{ 1,\ldots,i\} \backslash \bigcup_{t''< t'} I_{t''}$ and let $R := \{ 1,\ldots,i\} \backslash (G \cup \bigcup_{t''< t'} I_{t''})$
be the remaining row indices. 
The error that our rounding produces w.r.t. $i$ is 
\[
 \Big|\sum_{j \leq i} A_{j}(x-y)\Big| = \Big|\sum_{t''< t'} \underbrace{v_{I_{t''}}(x-y)}_{=0} + v_G(x-y) + v_R(x-y)\Big| 
\leq  4\sqrt{\ln\left( \frac{2}{\delta}\right)} \cdot \|v_G\|_2 + \underbrace{\|v_R\|_1}_{\leq 2C\sqrt{\delta}k}
\]
It remains to bound $\|v_G\|_2$. At this point, we crucially rely on the 
assumption $\|A_{i}\|_{\infty} \leq 2\delta \cdot \|A_{i}\|_1$. Using this together with Hölder's inequality and the triangle inequality we obtain
\[
\|v_G\|_2 \leq \sqrt{\|v_G\|_1 \cdot \|v_G\|_{\infty}} \leq \sqrt{\|v_{G}\|_1 \cdot \sum_{i \in G} \|A_i\|_{\infty} } 
\leq \sqrt{\|v_G\|_1 \cdot 2\delta \sum_{i \in G} \|A_i\|_1 } = \sqrt{2\delta} \underbrace{\|v_G\|_1}_{\leq O(k)} = O(\sqrt{\delta} k)
\]
Hence the claim is proven for small items. 
On the other hand for large items $i$ we do not even need to use the subgroups. 
Let $t'$ be the group with $i \in I_{t'}$ and denote $R := \{ 1,\ldots,i\} \backslash \bigcup_{t''<t'} I_{t''}$ as the remaining interval. 
Then we directly obtain $|\sum_{j \leq i} A_j(x-y)| = |\sum_{t''<t'} v_{I_{t''}}(x-y) + v_R(x-y)| \leq \|v_R\|_1 \leq O(\frac{1}{s_i})$. 
\end{proof}

One can alternatively prove Theorem~\ref{thm:ConstructivePartialColoringForBinPacking}
 by combining the classical partial coloring lemma 
and the Lovász-Spencer-Vesztergombi Theorem~\cite{DiscrepancyofSetSystemsAndMatrices-LSV86}.
Note that the bound of the classical partial coloring lemma  
involves an extra $O(\log \frac{1}{\lambda})$ term for $\lambda\leq2$. However, using the parametrization as
in \cite{DiscrepancyOfPermutations-SpencerEtAl} one could avoid loosing a super constant factor.
Moreover, instead of just 2 different group types (``groups'' and ``subgroups''), 
we could use an unbounded number to save the $\sqrt{\ln(\frac{2}{\delta})}$ factor. 
However, this would not improve on the overall approximation ratio.

Recall that for the gluing procedure in Lemma~\ref{lem:WellSpreadInstance} we need the 
property that the entries
in $x_p$ are not too tiny -- say at least $\frac{1}{\polylog(n)}$.
But this is easy to achieve (in fact, the next lemma also follows from 
\cite{KarmarkarKarp82} or \cite{EntropyMethodRothvoss-SODA12}).
\begin{lemma} \label{lem:RoundingToMultipleOfGamma}
Given any instance $1\geq s_1\geq\ldots\geq s_n\geq\frac{1}{n}$, $x \in \setR_{\geq0}^m$ and parameter $\gamma>0$.  Then one can compute a $y \succeq x$ in expected polynomial time
such that all $y_p$ are multiples of $\gamma$ and $(\bm{1},2s)^Ty \leq (\bm{1},2s)^Tx + O(\gamma \cdot \log^2 n)$.
\end{lemma}
\begin{proof}
After replacing $x$ with a
basic solution, we may assume that $|\supp(x)| \leq n$. We write $x = \gamma x' + \gamma z$ with $z \in \setZ_{\geq0}^m$ and $\bm{0} \leq x' \leq \bm{1}$. 
Now apply  $\log n$ times Theorem~\ref{thm:ConstructivePartialColoringForBinPacking}
with $\delta = 1$ to $x'$ to obtain $y' \in \{ 0,1\}^m$ with $\bm{1}^Ty' \leq \bm{1}^Tx'$ and
$|\sum_{j\leq i} A_j(x-y)| \leq O(\log n \cdot \frac{1}{s_i})$ (if at the end of the rounding process, the
fractional support goes below $100\log(n)$, we can stop and remove the remaining fractional patterns). Let $i_{\ell}$ be the largest item
in size class $\ell$. Then for $\ell = 0,\ldots,\log n$, we add
 $O(\log n \cdot 2^{\ell})$ copies of item $i_{\ell}$ to the waste of $y'$. Now
$(\bm{1},2s)^Ty' \leq (\bm{1},2s)^Tx' + O(\log^2 n)$ and $y' \succeq x'$. Eventually define $y := \gamma y' + \gamma z$
and observe that $y \succeq x$, all entries $y_p$ are multiples 
of $\gamma$ and $(\bm{1},2s)^Ty \leq (\bm{1},2s)^Tx + O(\gamma \log^2 n)$.
\end{proof}
Observe that just applying Lemma~\ref{lem:RoundingToMultipleOfGamma} with $\gamma = 1$ yields an integral 
solution with cost $OPT_f + O(\log^2 n)$.

\subsection{Proof of the main theorem}

It remains to put all ingredients together and show that each of 
the $\log n$ applications of the partial coloring lemma increases
the objective function by at most $O(\log \log n)$.
\begin{theorem} \label{thm:LogNLogLogNgap}
Let $1 \geq s_1 \geq \ldots \geq s_n \geq \frac{1}{n}$ with  and $x \in \setR_{\geq0}^m$ be given. Then there is an expected polynomial time algorithm to compute
$y\succeq x$ with $\bm{1}^Ty \leq \bm{1}^Tx + O(\log n\cdot\log \log n)$.
\end{theorem}
\begin{proof}
We choose $\delta := \beta := \gamma := \frac{1}{\lceil\log^4 n\rceil}$ and  $\varepsilon := \frac{1}{4\log^{12}n} \leq \frac{1}{2}\gamma \beta \delta$.
After moving to a basic solution and buying integral parts of $x$, we may assume that
$x \in [0,1]^m$ and $|\supp(x)| \leq n$. Recall that $i_{\ell} = \textrm{argmax}\{ i \mid 2^{-(\ell+1)} < s_i \leq 2^{-\ell}\}$. We perform the following algorithm:
\begin{enumerate}
\item[(1)] FOR $t = 1$ TO $\log n$ DO
  \begin{enumerate}
  \item[(2)] Apply Lemma~\ref{lem:RoundingToMultipleOfGamma} to have all $x_p$ being multiples of $\gamma$.
  \item[(3)] Apply Lemma~\ref{lem:WellSpreadInstance} to make $x$ $\delta$-well spread for items of size at most $\varepsilon$.
  \item[(4)] Apply Theorem~\ref{thm:ConstructivePartialColoringForBinPacking} to $x$ to halve the number 
of fractional entries\footnote{To be precise, we apply Theorem~\ref{thm:ConstructivePartialColoringForBinPacking} for the submatrix of $A$ corresponding to the columns in $\supp(x)$.}
(if $|\supp(x)| \leq 100\log(n)$, just set $x_p:=1$ for all patterns in the support).
  \item[(5)] FOR each size class $\ell$, add $O(2^{\ell})$ items of $i_{\ell}$ to the waste
if $s_{i_{\ell}} \geq \varepsilon$ and $O(\sqrt{\delta\ln(\frac{2}{\delta})} \cdot 2^{\ell})$ items if $s_{i_{\ell}}<\varepsilon$.
  \end{enumerate}
\item[(6)] Resubstitute glued items in $x$ and replace waste patterns via a greedy assignment 
to obtain $y$.
\end{enumerate}
Let $x^{(t)}$ be the value of $x$ at the end of the $t$th while loop. 
Let $x^{(t,2)}$ be the solution $x$ at the \emph{end} of step (2)
in the $t$th iteration. We define $x^{(t,3)},x^{(t,4)},x^{(t,5)}$ analogously.  
Furthermore, $x^{(t,1)}$ be $x$ at the beginning of the while loop. 

First of all, observe that $x^{(t,2)}$ satisfies   $x^{(t,2)} \succeq x^{(t,1)}$  and $x^{(t,2)}_p \in \gamma \setZ_{\geq0}$ for all $p$, by the properties of Lemma~\ref{lem:RoundingToMultipleOfGamma}.
The vector $x^{(t,3)}$ emerges from $x^{(t,2)}$ by grouping and gluing 
and according to Lemma~\ref{lem:WellSpreadInstance}, it satisfies
 $A_{ip} \leq \delta \cdot A_ix^{(t,3)} \leq \delta \cdot \sum_{\tilde{p} \in \supp(x^{(t,3)})} \tilde{p}_i$ for all $p \in \supp(x^{(t,3)})$ and all $i$ with 
$s_i \leq \varepsilon \leq \frac{1}{2}\gamma\beta\delta $.
Finally, the conditions of Theorem~\ref{thm:ConstructivePartialColoringForBinPacking}
are satisfied by parameters $\delta$ and $\varepsilon$, thus the extra items bought in step (5)
are enough to have $x^{(t,5)} \succeq x^{(t,3)}$. 
Note that none of the steps (2),(3),(5) increases the number of regular 
patterns in the support, but $|\{ p \mid x^{(t,4)}_p \notin \{ 0,1\}\}| \leq \frac{1}{2} |\{ p \mid x^{(t,3)}_p \notin \{ 0,1\}\}|$, 
thus $x^{(\log n)}$ is indeed integral.

Hence, by Corollary~\ref{cor:SequenceRoundingAndGluing} we know that $y$ will be 
a feasible solution to the original bin packing instance of cost at most $(\bm{1},2s)^Tx^{(\log n)} + 1$.
It remains to account for the increase in the objective function. 
Each application of (2) increases the objective function by at most $O(\gamma \log^2 n)$. 
(3) costs us $O(\beta \log^2 n)$ and $(4)+(5)$
increases the objective function by $O(\log \frac{1}{\varepsilon} + \sqrt{\delta\ln(\frac{2}{\delta})} \log n)$.
In total over $\log n$ iterations, 
\[
 \bm{1}^Ty - \bm{1}^Tx \leq O(\gamma \log^3 n) + O(\beta\log^3 n) + O(\log n\cdot\log \tfrac{1}{\varepsilon}) + O({\textstyle \sqrt{\delta\ln(\frac{2}{\delta})}}\log^2 n)
\leq O(\log n \cdot \log \log n)
\]
plugging in the choices for $\delta,\beta,\gamma,\varepsilon$.
\end{proof}

Together with the remark from Lemma~\ref{lem:AssumptionSizesAtLeast1-n},
the approximation guarantee for our main result, Theorem~\ref{thm:MainContribution} follows.
Let us conclude with a quick estimate on the running time. 
Given a Bin Packing instance $s_1,\ldots,s_n$ (with one copy of each item, so $n$ is the total number of items), 
one can compute a fractional solution $x$ of cost $\bm{1}^Tx \leq (1+\varepsilon)OPT_f$ in time $O((\frac{n^2}{\varepsilon^4}+\frac{1}{\varepsilon^6}) \log^5(\frac{n}{\varepsilon}))$
with $|\textrm{supp}(x)| \leq n$ (see Theorem~5.11 in \cite{FractionalPackingAndCovering-PlotkinShmoysTardos-Journal95}).
We set $\varepsilon := \frac{1}{n}$ and obtain an $x$ with $\bm{1}^Tx \leq OPT_f + 1$ in time $O(n^6 \log^5(n))$.
It suffices to run the Constructive Partial Coloring Lemma with error parameter $\delta := \frac{1}{n}$, 
which takes time $O(\frac{\tilde{n}^3}{\delta^2} \log \frac{\tilde{n}\tilde{m}}{\delta}) = \tilde{O}(n^5)$ 
where $\tilde{n} \leq n \cdot \textrm{polylog}(n)$ is the number of vectors and $\tilde{m} \leq n$ is the dimension
of $x$.
In other words,  the running time is dominated by the computation of the fractional solution.
Finally we obtain a vector $y$ with $y_p \in [0,\frac{1}{n}] \cup [1-\frac{1}{n},1]$ and $|\supp(y)| \leq n$.
 We move those entries with $y_p \leq \frac{1}{n}$ to the waste, increasing the objective function by at most $n \cdot \frac{1}{n} = 1$
and we roundup those entries with $y_p \geq 1-\frac{1}{n}$.


\section{Remarks}

An interesting observation concerning the application of the 
Constructive Partial Coloring Lemma is the following: recall that we gave vectors $v_I$ for 
groups and vectors $v_G$ for all subgroups as input to Lemma~\ref{lem:ConstructivePartialColoring}.
But the solution $y$ returned by that lemma is the end point of a Brownian motion 
and satisfies $\Pr[|v(x-y)| \leq \lambda \|v\|_2] \leq e^{-\Omega(\lambda^2)}$ for every $\lambda \geq 0$ and $v \in \setR^m$
regardless whether $v$ is known to the algorithm or not. 
If we choose $\lambda_G := \log n$ (and $\delta,\gamma,\varepsilon$ slightly more generous), then
the guarantee $|v_G(x-y)| \leq \lambda_G\|v_G\|$ is satisfied for all subgroups with high probability
anyway and there is no need to include them in the input.

Moreover, we are not even using the full power 
of the constructive partial coloring lemma. Suppose we had only a 
weaker Lemma~\ref{lem:ConstructivePartialColoring} which needs the
stronger assumption that for example $\sum_i (1+\lambda_i)^{-10} \leq \frac{m}{16}$
instead of the exponential decay in $\lambda$. We would still obtain
the same asymptotic bound of $O(\log n \cdot \log \log n)$, thus a simple fine tuning of
parameters is not going to give any improvement.

The obvious question is, how tight is our analysis? 
In fact, it is plausible that a slightly changed algorithm with a more careful analysis 
can reduce the gap from  $O(\log n \cdot \log \log n)$ to $O(\log n)$. 
On the other hand, consider the seemingly simple \emph{3-Partition} case in which all $n$ items
have size $\frac{1}{4} < s_i < \frac{1}{2}$. Both, the approaches of
Karmarkar and Karp~\cite{KarmarkarKarp82} and ours provide 
a $O(\log n)$ upper bound on the integrality gap.
The construction of Newman and Nikolov~\cite{CounterexampleBecksPermutationConjecture-FOCS12} 
of 3 badly colorable permutations
can be used to define a 3-Partition instance with an optimum fractional
solution $x \in \{ 0, \frac{1}{2} \}^m$ such that any integral solution $y \in \setZ_{\geq0}^m$
with $\supp(y) \subseteq \supp(x)$ and $\bm{1}^Ty - \bm{1}^Tx \leq o(\log n)$ 
satisfies $\max_{i \in [n]} \{ \sum_{j \leq i} (A_jx-A_jy) \} \geq \Omega(\log n)$.
This suggests that either the $\log n$ bound is best possible for 3-Partition 
or some fundamentally new ideas are needed to make progress.

\paragraph{Acknowledgements.}

The author is grateful to Michel X. Goemans for helpful discussions and support
and to Nikhil Bansal for reading a preliminary version. 

\bibliographystyle{alpha}
\bibliography{BinPackingImprovement}

\end{document}